\documentclass{article}
\usepackage[dvipdfmx]{graphicx}
\usepackage[dvipdfmx]{hyperref}
\usepackage{arxiv}
\usepackage[utf8]{inputenc} 
\usepackage[T1]{fontenc}    
\usepackage{url}            
\usepackage{booktabs}       
\usepackage{amsfonts}       
\usepackage{nicefrac}       
\usepackage{microtype}      
\usepackage{bm}
\usepackage{amsmath}
\usepackage{amsthm}
\usepackage{algorithm}
\usepackage{algorithmic}
\usepackage{comment}
\usepackage{multirow}

\title{An Overflow/Underflow-Free Fixed-Point Bit-Width \\Optimization Method for OS-ELM Digital Circuit}
\author{
  Mineto Tsukada\\
  Keio University\\
  3-14-1 Hiyoshi, Kohoku-ku, Yokohama, Japan\\
  \texttt{tsukada@arc.ics.keio.ac.jp}\\
  \And
  Hiroki Matsutani \\
  Keio University\\
  3-14-1 Hiyoshi, Kohoku-ku, Yokohama, Japan\\
  \texttt{matutani@arc.ics.keio.ac.jp} \\
}

\theoremstyle{definition}
\newtheorem{theorem}{Theorem}

\begin{document}
\maketitle

\begin{abstract}
Currently there has been increasing demand for real-time training on resource-limited IoT devices such as smart sensors,
which realizes standalone online adaptation for streaming data without data transfers to remote servers.
OS-ELM (Online Sequential Extreme Learning Machine) has been one of promising neural-network-based online algorithms for on-chip learning
because it can perform online training at low computational cost and is easy to implement as a digital circuit.
Existing OS-ELM digital circuits employ fixed-point data format and the bit-widths are often manually tuned, however, this may cause overflow or underflow
which can lead to unexpected behavior of the circuit.
For on-chip learning systems, an overflow/underflow-free design has a great impact since
online training is continuously performed and the intervals of intermediate variables will dynamically change as time goes by.
In this paper, we propose an overflow/underflow-free bit-width optimization method for fixed-point digital circuit of OS-ELM.
Experimental results show that our method realizes overflow/underflow-free OS-ELM digital circuits
with 1.0x - 1.5x more area cost compared to an ordinary simulation-based optimization method
where overflow or underflow can happen.
\end{abstract}



\section{Introduction}\label{sec:intro}
Currently there has been increasing demand for real-time training on resource-limited IoT devices (e.g. smart sensors and micro computers),
which realizes standalone online adaptation for streaming data without transferring data to remote servers,
and avoids additional power consumption for communication \cite{os_elm_fpga_tsukada_tc}.
OS-ELM (Online Sequential Extreme Learning Machine) \cite{os_elm} has been one of promising neural-network-based online algorithms for on-chip learning
because it can perform online training at low computational cost and is easy to implement as a digital circuit.
Several papers have proposed design methodologies and implementations of OS-ELM digital circuits
and shown that OS-ELM can be implemented in a small-size FPGA and still be able to perform online training in less than one millisecond \cite{os_elm_fpga_tsukada_tc, os_elm_fpga_tsukada, os_elm_fpga_villora, os_elm_fpga_safaei}.

Existing OS-ELM digital circuits often employ fixed-point data format and the bit-widths are manually tuned according
to the requirements (e.g. resource and timing constraints), however, this may cause overflow or underflow
which can lead to unexpected behavior of the circuit.
A lot of works have proposed bit-width optimization methods that analytically derive the lower and upper bounds of intermediate variables
and automatically optimize the fixed-point data format, ensuring that overflow/underflow never happens \cite{aawlo_lee, smtwlo_kinsman, hrwlo_boland}.
For on-chip learning systems, an overflow/underflow-free design has a significant impact because
online training is continuously performed and the intervals of intermediate variables will dynamically change as time goes by.

In this paper we propose an overflow/underflow-free bit-width optimization method for fixed-point OS-ELM digital circuits.
This work makes the following contributions.
\begin{itemize}
    \item We propose an interval analysis method for OS-ELM using affine arithmetic \cite{aa},
    one of the most widely-used interval arithmetic models.
    Affine arithmetic has been used in a lot of existing works for determining optimal integer bit-widths that never cause overflow and underflow.
    \item In affine arithmetic, division is
    defined only if the denominator does not include zero;
    otherwise the algorithm cannot be represented in affine arithmetic.
    OS-ELM's training algorithm contains one division and we analytically prove that the denominator does not include zero.
    Based on this proof, we also propose a mathematical trick to safely represent OS-ELM in affine arithmetic.
    \item Affine arithmetic can represent only fixed-length computation graphs and
    unbounded loops are not supported in affine arithmetic.
    However, OS-ELM's training algorithm is an iterative algorithm
    where current outputs are used as the next inputs endlessly.
    We propose an empirical solution for this problem based on simulation results, and
    verify its effectiveness in the evaluation section.
    \item We evaluate the performance of our interval analysis method,
    using an fixed-point IP core called OS-ELM Core to demonstrate the practicality of our method.
\end{itemize}

The rest of this paper is organized as follows;
Section \ref{sec:prelim} gives a brief introduction of basic technologies behind this work.
Our interval analysis method is proposed in Section \ref{sec:method}.
Section \ref{sec:impl} briefly describes the design of OS-ELM Core.
The proposed interval analysis method is evaluated in Section \ref{sec:eval}.
Section \ref{sec:conc} concludes this paper.
Please refer to Table \ref{tb:notation} and Table \ref{tb:desc}
for the notation rules and the description description of special variables that frequently appear in this paper.
\section{Preliminaries}\label{sec:prelim}
\subsection{ELM}\label{sec:prelim.elm}
We first introduce ELM (Extreme Learning Machine) \cite{elm} prior to OS-ELM.
ELM illustrated in Figure \ref{fig:elm} is a neural-network-based model that consists of an input layer, one hidden layer, and an output layer.
If an $n$-dimensional input $\bm{X} \in \mathbb{R}^{k \times n}$ of batch size = $k$ is given,
the $m$-dimensional prediction output $\bm{Y} \in \mathbb{R}^{k \times m}$ can be computed in the following formula.
\begin{equation}\label{eq:elm.predict}
    \bm{Y} = \mathrm{G}(\bm{X} \cdot \bm{\alpha} + \bm{b})\bm{\beta}
\end{equation}
\begin{figure}[b]
    \centering
    \includegraphics[width=85.0mm]{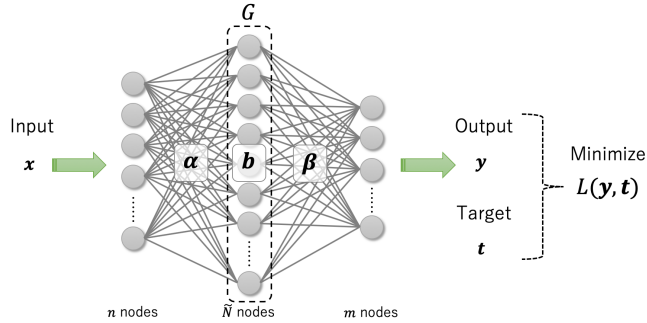}
    \caption{
        Extreme learning machine.
        $n$/$\tilde{N}$/$m$ represents the number of input/hidden/output nodes.
        $\bm{\alpha} \in \mathbb{R}^{n \times \tilde{N}}$ is a weight matrix connecting
        the input and the hidden layers.
        $\bm{\beta} \in \mathbb{R}^{\tilde{N} \times m}$ is another weight matrix connecting the hidden and output layers.
        $\bm{b} \in \mathbb{R}^{1 \times \tilde{N}}$ is a bias vector of the hidden layer,
        and $\mathrm{G}$ is an activation function applied to the hidden layer output.
        $\bm{\alpha}$ and $\bm{b}$ are non-trainable constant parameters initialized with random values.
    }
    \label{fig:elm}
\end{figure}

ELM uses a finite number of input-target pairs for training.
Suppose an ELM model can approximate $N$ input-target pairs $\{\bm{X} \in \mathbb{R}^{N \times n}, \bm{T} \in \mathbb{R}^{N \times m}\}$ with zero error,
it implies there exists $\bm{\beta}$ that satisfies the following equation.
\begin{equation}\label{eq:elm_zero_error}
    \mathrm{G}(\bm{X} \cdot \bm{\alpha} + \bm{b})\bm{\beta} = \bm{T}
\end{equation}
Let $\bm{H} \equiv \mathrm{G}(\bm{X} \cdot \bm{\alpha} + \bm{b})$ then
the optimal solution $\bm{\beta}^{\ast}$ is derived with the following formula.
\begin{equation}\label{eq:elm_train}
    \bm{\beta}^{\ast} = \bm{H}^{\dagger}\bm{T}
\end{equation}
$\bm{H}^{\dagger}$ is the pseudo inverse of $\bm{H}$.
The whole training process finishes by replacing $\bm{\beta}$ with $\bm{\beta}^{\ast}$.
ELM takes one-shot optimization approach unlike backpropagation-based neural-networks (BP-NNs),
which makes the whole training process faster.
ELM is known to finish optimization process faster than BP-NNs \cite{elm}.

\subsection{OS-ELM}\label{sec:prelim.os_elm}
ELM is a batch learning algorithm; ELM needs to re-train with the whole training dataset, including training samples already learned in the past,
in order to learn new training samples.
OS-ELM \cite{os_elm} is an ELM variant that can perform online learning instead of batch learning.
Suppose the $i$th training samples $\{\bm{X}_i \in \mathbb{R}^{k_i \times n}, \bm{T}_i \in \mathbb{R}^{k_i \times m}\}$ of batch size $= k_i$ is given,
OS-ELM computes the $i$th optimal solution $\bm{\beta}_i$ in the following formula.
\begin{equation}\label{eq:os_elm.train.batch}
    \begin{split}
        \bm{P}_i &= \bm{P}_{i-1} - \bm{P}_{i-1}\bm{H}_i^T(\bm{I} + \bm{H}_i\bm{P}_{i-1}\bm{H}_i^T)^{-1}\bm{H}_i\bm{P}_{i-1}\\
        \bm{\beta}_i &= \bm{\beta}_{i-1} + \bm{P}_i\bm{H}_i^T(\bm{T}_i - \bm{H}_i\bm{\beta}_{i-1}),
    \end{split}
\end{equation}
where $\bm{H}_i \equiv G(\bm{X}_i \cdot \bm{\alpha} + \bm{b})$.
$\bm{P}_0$ and $\bm{\beta}_0$ are computed as follows.
\begin{equation}\label{eq:os_elm.init}
    \begin{split}
        \bm{P}_0 &= (\bm{H}_0^T\bm{H}_0)^{-1}\\
        \bm{\beta}_0 &= \bm{P}_0\bm{H}_0^T\bm{T}_0
    \end{split}
\end{equation}
Note that OS-ELM and ELM produce the same solution as long as the training dataset is the same.

Specially, when $k_{i} = 1$ Equation \ref{eq:os_elm.train.batch} can be simplified into
\begin{equation}\label{eq:os_elm.train}
    \begin{split}
        \bm{P}_i &= \bm{P}_{i-1} - \frac{\bm{P}_{i-1}\bm{h}_i^T\bm{h}_i\bm{P}_{i-1}}{1+\bm{h}_i\bm{P}_{i-1}\bm{h}_i^T} \\
        \bm{\beta}_i &= \bm{\beta}_{i-1} + \bm{P}_i\bm{h}_i^T(\bm{t}_i - \bm{h}_i\bm{\beta}_{i-1}),
    \end{split}
\end{equation}
where $\bm{h}_i \equiv G(\bm{x}_i \cdot \bm{\alpha} + \bm{b})$.
Note that $\bm{x}_i$/$\bm{t}_i$/$\bm{h}_i$ is a special case of $\bm{X}_i$/$\bm{T}_i$/$\bm{H}_{i}$ when $k_{i} = 1$.
Equation \ref{eq:os_elm.train} is more costless than Equation \ref{eq:os_elm.train.batch} in terms of computational complexity since
a costly matrix inverse $(\bm{I} + \bm{H}_i\bm{P}_{i-1}\bm{H}_i^T)^{-1}$ has been replaced with a simple reciprocal operation $\frac{1}{1+\bm{h}_i\bm{P}_{i-1}\bm{h}_i^T}$ \cite{os_elm_fpga_tsukada_tc}.
In this work we refer to Equation \ref{eq:os_elm.train} as ``training algorithm'' of OS-ELM.
Equation \ref{eq:os_elm.init} is referred to as ``initialization algorithm''.

The prediction algorithm is below.
\begin{equation}\label{eq:os_elm.predict}
    \bm{y} = \mathrm{G}(\bm{x} \cdot \bm{\alpha} + \bm{b})\bm{\beta},
\end{equation}
where $\bm{y}$ is a special case of $\bm{Y}$ when the batch size is equal to 1.
We refer to Equation \ref{eq:os_elm.predict} as ``prediction algorithm'' of OS-ELM.

\subsection{Interval Analysis}\label{sec:prelim.ra}
To realize an overflow/underflow-free fixed-point design you need to know the interval of each variable
and allocate sufficient integer bits that never cause overflow and underflow.
Existing interval analysis methods for fixed-point design are categorized into a (1) dynamic method or a (2) static method \cite{fixed_point_refine}.
Dynamic methods \cite{wlo_dynamic_cmar, wlo_dynamic_gaffar, wlo_dynamic_keding, wlo_dynamic_shi} often take a simulation-based approach with tons of test inputs.
It is known that dynamic methods often produce a better result close to the true interval compared to static methods,
although they tend to take a long time due to exhaustive search and
may encounter overflow or underflow if unseen inputs are found in runtime.
Static methods \cite{aawlo_lee, aawlo_cong, aawlo_vakili, smtwlo_kinsman, hrwlo_boland}, on the other hand, take a more analytical approach;
they often involve solving equations and deriving upper and lower bounds of each variable without test inputs.
Static methods produce a more conservative result (i.e. a wider interval) compared to dynamic methods, although the result is analytically guaranteed.
In this work we employ a static method for interval analysis as
the goal is to realize an overflow/underflow-free fixed-point OS-ELM digital circuit with analytical guarantee.

Interval arithmetic (IA) \cite{ia} is one of the oldest static interval analysis methods.
In IA every variable is represented as an interval $[x_1, x_2]$
where $x_1$ and $x_2$ are the lower and upper bounds of the variable.
Basic operations $\{+, -, \times\}$ are defined as follows;
\begin{equation}\label{eq:ia}
    \begin{split}
        [x_1, x_2] + [y_1, y_2] &= [x_1 + y_1, x_2 + y_2]\\
        [x_1, x_2] - [y_1, y_2] &= [x_1 - y_2, x_2 - y_1]\\
        [x_1, x_2] \times [y_1, y_2] &= [\\
            &\mathrm{min}(x_1y_1, x_1y_2, x_2y_1, x_2y_2),\\
            &\mathrm{max}(x_1y_1, x_1y_2, x_2y_1, x_2y_2)]
    \end{split}
\end{equation}
IA guarantees intermediate intervals as long as input intervals are known.
However, IA suffers from the \textit{dependency problem};
for example, $x - x$ where $x \in [x_1, x_2]\;(x_1 < x_2)$ should be 0
in ordinary algebra, although the result in IA is $[x_1 - x_2, x_2 - x_1]$,
a wider interval than the true tightest range $[0, 0]$, which
makes subsequent intervals get wider and wider.
The cause of this problem is that IA ignores
correlation of variables; $x - x$ is treated a self-subtraction in ordinary algebra
but it is regarded as a subtraction between independent intervals in IA.

Affine arithmetic (AA) \cite{aa} is a refinement of IA proposed by Stolfi \textit{et. al}.
AA keeps track of correlation of variables and is known to
obtain tighter bounds close to the true range compared to IA.
AA has been applied into a lot of fixed-point/floating-point bit-width optimization systems \cite{aawlo_lee, aawlo_fang, aawlo_cong, aawlo_pu}
and still widely used in recent works \cite{aawlo_vakili, aawlo_wang, aawlo_bellal}.
We use AA throughout this work.

\begin{figure}[b]
    \centering
    \includegraphics[width=45.0mm]{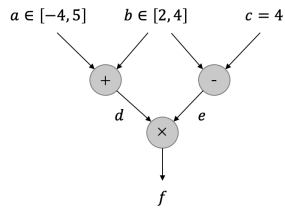}
    \caption{
        A simple tutorial of AA.
        In AA all input intervals ($a, b, c$ in this tutorial) must be known.
        Affine forms of $a, b, c, d, e, f$ are computed as follows; $\hat{a} = 0.5 + 4.5\epsilon_{a}\:,\hat{b} = 3.0 + \epsilon_{b}\;,\hat{c} = 4.0\;,\hat{d} = 3.5 + 4.5\epsilon_{a} + \epsilon_{b}\;,\hat{e} = -1.0 + \epsilon_{b}\;,\hat{f} = -3.5 - 4.5\epsilon_{a} + 2.5\epsilon_{b} + 5.5\epsilon_{f}$.
        We can derive $\mathrm{interval}(\hat{d}) = [-2.0, 9.0]$, $\mathrm{interval}(\hat{e}) = [-2.0, 0.0]$, and $\mathrm{interval}(\hat{f}) = [-16.0, 9.0]$ from Equation \ref{eq:affine.interval}.
    }
    \label{fig:graph_example}
\end{figure}
\subsection{Affine Arithmetic}\label{sec:prelim.aa}
In AA the interval of variable $x$ is represented in an \textit{affine form } $\hat{x}$ given by;
\begin{equation}\label{eq:affine}
    \hat{x} = x_0 + x_1\epsilon_1 + x_2\epsilon_2 + \cdots + x_{n}\epsilon_{n},
\end{equation}
where $\epsilon_{i} \in [-1, 1]$. $x_{i}$ is a coefficient and $\epsilon_{i}$ is an \textit{uncertainty variable} which takes $[-1, 1]$;
an affine form is a linear combination of uncertainty variables.

The interval of $\hat{x}$ can be computed as below.
\begin{equation}\label{eq:affine.interval}
    \begin{split}
        \mathrm{interval}(\hat{x}) &= [\mathrm{inf}(\hat{x}), \mathrm{sup}(\hat{x})] \\
        \mathrm{inf}(\hat{x}) &= x_0 - \sum_{i}{|x_{i}|}\\
        \mathrm{sup}(\hat{x}) &= x_0 + \sum_{i}{|x_{i}|}
    \end{split}
\end{equation}
$\mathrm{inf}(\hat{x})$ computes the lower bound of $\hat{x}$ and $\mathrm{sup}(\hat{x})$ is the upper bound.
Conversely a variable that ranges $[a, b]$ can be converted into an affine form $\hat{x} = x_0 + x_1\epsilon_{1}$ with
\begin{equation}
    x_0 = \frac{b + a}{2},\; x_1 = \frac{b - a}{2}.
\end{equation}

Addition/subtraction between affine forms $\hat{x}$ and $\hat{y}$ is simply defined as $\hat{x} \pm \hat{y} = (x_0 \pm y_0) + \sum_{i}{(x_i \pm y_i)\epsilon_i}$.
However, multiplication $\hat{x} \ast \hat{y}$ is a little bit complicated.
\begin{equation}
    \begin{split}
        \hat{x} \ast \hat{y} &= x_0y_0 + \sum_{i}{(x_0y_i + y_0x_i)\epsilon_i} + Q\\
        Q &= \sum_{i}{(x_i\epsilon_i)}\sum_{i}{(y_i\epsilon_i)}
    \end{split}
\end{equation}
Note that $Q$ is not an affine form (i.e. $Q$ is not a linear combination of $\epsilon_{i}$) and it needs approximation to become an affine form.
A conservative approximation shown below is often taken \cite{aawlo_lee, aawlo_cong, aawlo_vakili}.
\begin{equation}
    Q \approx uv\epsilon_{\ast},\; u = \sum_{i}{|x_{i}|},\; v = \sum_{i}{|y_{i}|},
\end{equation}
where $\epsilon_{\ast} \in [-1, 1]$ is a new uncertainty variable.
Note that $uv\epsilon_{\ast} \ge \sum_{i}{(x_i\epsilon_i)}\sum_{i}{(y_i\epsilon_i)}$.
See Figure \ref{fig:graph_example} for a simple tutorial of AA.

Division $\hat{z} = \frac{\hat{x}}{\hat{y}}$ is often separated into $\hat{x} \ast \frac{1}{\hat{y}}$.
There are mainly two approximation methods to compute $\frac{1}{\hat{y}}$: (1) the min-max approximation and
(2) the chebyshev approximation. Here we show the definition of $\frac{1}{\hat{y}}$ with the min-max approximation.
\begin{equation}
    \begin{split}
    p &= \left\{
        \begin{array}{ll}
            -\frac{1}{b^2} & (\mathrm{if}\;b > a > 0) \\
            -\frac{1}{a^2} & (\mathrm{if}\;0 > b > a)
        \end{array}\right. \\
    q &= \frac{(a + b)^2}{2ab^2}\;,d = \frac{(a - b)^2}{2ab^2} \\
    \frac{1}{\hat{y}} &= (p \cdot y_0 + q) + \sum_{i}{p \cdot (y_i\epsilon_i)} + d\epsilon_{\ast},
    \end{split}
\end{equation}
where $a = \mathrm{inf}(\hat{y})$ and $b = \mathrm{sup}(\hat{y})$.
Note that $\frac{1}{\hat{y}}$ is defined only if $b > a > 0$ or $0 > b > a$.
The denominator $\hat{y}$ must not include zero.

\subsection{Determination of Integer Bit-Width}\label{sec:prelim.optim}
Suppose we have an affine form $\hat{x}$, the minimum number of integer bits that never cause overflow and underflow is computed by;
\begin{equation}\label{eq:optim}
    \begin{split}
    IB &= \lceil \log_2 (\mathrm{max}(|\mathrm{inf}(\hat{x})|, |\mathrm{sup}(\hat{x})|) + 1) + \alpha, \\
    \alpha &= \left\{
        \begin{array}{ll}
            1 & (\mathrm{if}\;\mathrm{signed}) \\
            0 & \mathrm{else.}
        \end{array}\right.
    \end{split}
\end{equation}
$IB$ represents the optimal integer bit-width.

\begin{algorithm}[t]
    \caption{
        $\mathrm{T}(\bm{x}_i, \bm{t}_i, \bm{\alpha}, \bm{b}, \bm{P}_{i-1}, \bm{\beta}_{i-1}) \mapsto \{\bm{P}_i, \bm{\beta}_i\}$ ($1 \le i \le N$).
    }
    \begin{algorithmic}[1]\label{alg:graph.t}
        \REQUIRE $\bm{x}_i, \bm{t}_i, \bm{\alpha}, \bm{b}, \bm{P}_{i-1}, \bm{\beta}_{i-1}$
        \ENSURE $\bm{h}_i = \bm{x}_i \cdot \bm{\alpha} + \bm{b}$,\\$\bm{P}_i = \bm{P}_{i-1} - \frac{\bm{P}_{i-1}\bm{h}_i^T\bm{h}_i\bm{P}_{i-1}}{1+\bm{h}_i\bm{P}_{i-1}\bm{h}_i^T}$,\\$\bm{\beta}_i = \bm{\beta}_{i-1} + \bm{P}_i\bm{h}_i^T(\bm{t}_i - \bm{h}_i\bm{\beta}_{i-1})$
        \STATE $\bm{e}_i \leftarrow \bm{x}_i \cdot \bm{\alpha}$ 
        \STATE $\bm{h}_i \leftarrow \bm{e}_i + \bm{b}$ 
        \STATE $\bm{\gamma}^{(1)}_i \leftarrow \bm{P}_{i-1} \cdot \bm{h}_i^T$
        \STATE $\bm{\gamma}^{(2)}_i \leftarrow \bm{h}_i \cdot \bm{P}_{i-1}$
        \STATE $\bm{\gamma}^{(3)}_i \leftarrow \bm{\gamma}^{(1)}_i \cdot \bm{\gamma}^{(3)}_i$ 
        \STATE $\gamma^{(4)}_i \leftarrow \bm{\gamma}^{(2)}_i \cdot \bm{h}_i^T$ 
        \STATE $\gamma^{(5)}_i \leftarrow \gamma^{(4)}_i + 1$ 
        \STATE $\bm{\gamma}^{(6)}_i \leftarrow \bm{\gamma}^{(3)}_i / \gamma^{(5)}_i$ 
        \STATE $\bm{P}_i \leftarrow \bm{P}_i - \bm{\gamma}^{(6)}_i$ 
        \STATE $\bm{\gamma}^{(7)}_i \leftarrow \bm{P}_i \cdot \bm{h}_i^T$
        \STATE $\bm{\gamma}^{(8)}_i \leftarrow \bm{h}_i \cdot \bm{\beta}_{i-1}$
        \STATE $\bm{\gamma}^{(9)}_i \leftarrow \bm{t}_i - \bm{\gamma}^{(8)}_i$ 
        \STATE $\bm{\gamma}^{(10)}_i \leftarrow \bm{\gamma}^{(7)}_i \cdot \bm{\gamma}^{(9)}_i$ 
        \STATE $\bm{\beta}_i \leftarrow \bm{\beta}_{i-1} + \bm{\gamma}^{(10)}_i$ 
        \RETURN $\{\bm{P}_i, \bm{\beta}_i\}$
    \end{algorithmic}
\end{algorithm}
\begin{algorithm}[t]
    \caption{$\mathrm{P}(\bm{x}, \bm{\alpha}, \bm{b}, \bm{\beta}) \mapsto \bm{y}$}
    \begin{algorithmic}[1]\label{alg:graph.p}
        \REQUIRE $\bm{x}, \bm{\alpha}, \bm{b}, \bm{\beta}$
        \ENSURE $\bm{y} = (\bm{x} \cdot \bm{\alpha} + \bm{b})\bm{\beta}$
        \STATE $\bm{e} \leftarrow \bm{x} \cdot \bm{\alpha}$
        \STATE $\bm{h} \leftarrow \bm{e} + \bm{b}$ 
        \STATE $\bm{y} \leftarrow \bm{h} \cdot \bm{\beta}$ 
        \RETURN $\bm{y}$
    \end{algorithmic}
\end{algorithm}

\section{AA-Based Interval Analysis for OS-ELM}\label{sec:method}
In this section we propose the AA-based interval analysis method for OS-ELM.
The process is two-fold: \textcircled{\scriptsize{1}} Build the computation graph equivalent to OS-ELM.
\textcircled{\scriptsize{2}} Compute the affine form and interval for every variable existing in OS-ELM, using Equation \ref{eq:affine.interval}.
Figure \ref{fig:graph} shows computation graphs for OS-ELM.
``Training graph'' corresponds to the training algorithm (Equation \ref{eq:os_elm.train}), and ``prediction graph'' corresponds to the prediction algorithm (Equation \ref{eq:os_elm.predict}).
\begin{figure}[t]
    \centering
    \includegraphics[width=95.0mm]{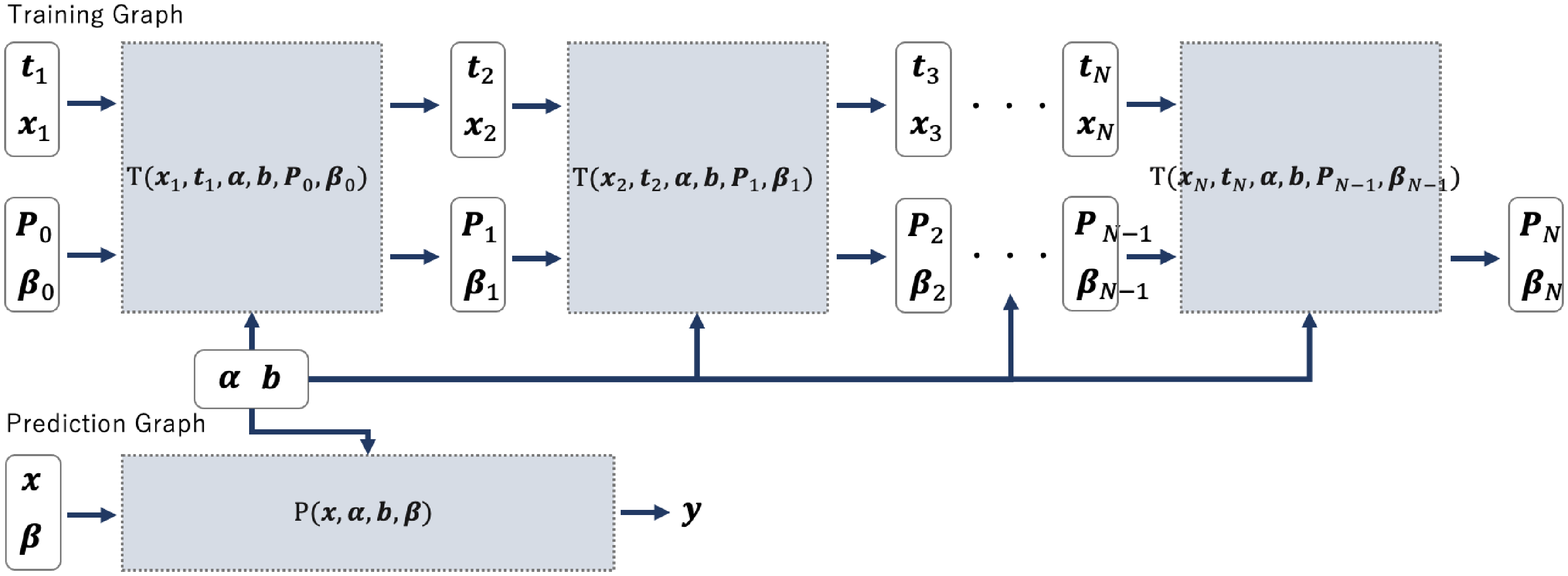}
    \caption{Computation graphs for OS-ELM. $N$ represents the total number of training steps.}
    \label{fig:graph}
\end{figure}
\begin{figure*}[t]
    \centering
    \includegraphics[width=170.0mm]{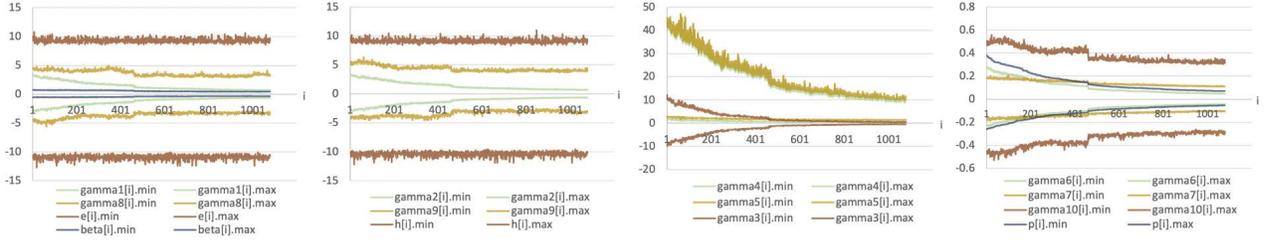}
    \caption{
        Observed intervals of $\{\bm{\gamma}^{(1)}_i, \ldots, \bm{\gamma}^{(10)}_i, \bm{P}_i, \bm{\beta}_i, \bm{e}_i, \bm{h}_i\}$ ($1 \le i \le N = 1,079$) on Digits dataset.
        The x-axis represents the training step $i$, and the y-axis plots the observed intervals (the maximum and minimum values) of each variable at training step $i$.
    }
    \label{fig:digits}
\end{figure*}

$\mathrm{T}(\bm{x}_i, \bm{t}_i, \bm{\alpha}, \bm{b}, \bm{P}_{i-1}, \bm{\beta}_{i-1}) \mapsto \{\bm{P}_i, \bm{\beta}_i\}$
defined in Algorithm \ref{alg:graph.t} represents a sub-graph that computes a single iteration of the OS-ELM training algorithm.
Training graph concatenates $N$ sub-graphs, where $N$ is the total number of training steps.
Training graph takes $\{\bm{x}_1, \ldots, \bm{x}_N, \bm{t}_1, \ldots, \bm{t}_N, \bm{\alpha}, \bm{b}, \bm{P}_0, \bm{\beta}_0\}$ as input and outputs $\{\bm{P}_N, \bm{\beta}_N\}$.
$\mathrm{P}(\bm{x}, \bm{\alpha}, \bm{b}, \bm{\beta}) \mapsto \bm{y}$ defined in Algorithm \ref{alg:graph.p} represents prediction graph.
Prediction graph takes $\{\bm{x}, \bm{\alpha}, \bm{b}, \bm{\beta}\}$ as input and outputs $\bm{y}$.

The goal is to obtain the intervals of $\{\bm{\gamma}_i^{(1)}, \ldots, \bm{\gamma}_i^{(10)}$, $\bm{P}_i$, $\bm{\beta}_i$, $\bm{e}_i, \bm{h}_i\}$ ($1 \le i \le N$) for training graph
and $\{\bm{e}, \bm{h}, \bm{y}\}$ for prediction graph, through AA.
In this paper, the interval of a matrix $\bm{A} \in \mathbb{R}^{u \times v}$ is computed as follows.
\begin{equation}\label{eq:affine.matrix_interval}
    \begin{split}
        \mathrm{interval}(\hat{\bm{A}}) &= [\mathrm{inf}(\hat{\bm{A}}), \mathrm{sup}(\hat{\bm{A}})] \\
        \mathrm{inf}(\hat{\bm{A}}) &= \mathrm{min}(\mathrm{inf}(\hat{A}_{[0,0]}), \ldots, \mathrm{inf}(\hat{A}_{[u-1,v-1]}))\\
        \mathrm{sup}(\hat{\bm{A}}) &= \mathrm{max}(\mathrm{sup}(\hat{A}_{[0,0]}), \ldots, \mathrm{sup}(\hat{A}_{[u-1,v-1]})),
    \end{split}
\end{equation}
where $\hat{\bm{A}}$ is the affine form of $\bm{A}$, and $\hat{A}_{[i,j]}$ is the $ij$ element of $\hat{\bm{A}}$.

\subsection{Constraints}
Remember that all input intervals must be known in AA; in other words the intervals of $\{\bm{x}_1, \ldots, \bm{x}_N$, $\bm{t}_1, \ldots, \bm{t}_N$, $\bm{\alpha}$, $\bm{b}$, $\bm{P}_0$, $\bm{\beta}_0\}$ for training graph
and $\{\bm{x}, \bm{\alpha}, \bm{b}, \bm{\beta}\}$ for prediction graph must be given.
In this work we assume that the intervals of $\{\bm{x}, \bm{x}_1, \ldots, \bm{x}_N, \bm{t}_1, \ldots, \bm{t}_N\}$ are $[0, 1]$,
and those of $\{\bm{\alpha}, \bm{b}\}$ are $[-1, 1]$.
$\{\bm{P}_0, \bm{\beta}_0\}$ is computed by Equation \ref{eq:os_elm.init}.
The interval of $\bm{\beta}$ (an input of prediction graph) is computed in the way described in Section \ref{sec:method.predict_graph}.

\subsection{Interval Analysis for Training Graph}\label{sec:method.train_graph}
The goal of training graph is to find the intervals of $\{\bm{\gamma}_i^{(1)}, \ldots, \bm{\gamma}_i^{(10)}$, $\bm{P}_i$, $\bm{\beta}_i$, $\bm{e}_i, \bm{h}_i\}$ for $1 \le i \le N$,
however, we have to deal with a critical problem; OS-ELM is an online learning algorithm and the total number of training steps $N$ is unknown
as training may occur in runtime (i.e. $N$ can increase in runtime).
if $N$ is unknown, the training graph grows endlessly and interval analysis becomes infeasible.
We need to determine a ``reasonable'' value of $N$ for training graph.

\subsubsection{Determination of $N$}\label{sec:method.train_analysis.n}
To determine $N$, we conducted an experiment to analyze
the intervals of $\{\bm{\gamma}_i^{(1)}, \ldots, \bm{\gamma}_i^{(10)}$, $\bm{P}_i$, $\bm{\beta}_i$, $\bm{e}_i, \bm{h}_i\}$
for $1 \le i \le N$. The procedure is as follows:
\textcircled{\scriptsize{1}} Implement OS-ELM's initialization and training algorithms in double-precision format.
\textcircled{\scriptsize{2}} Compute initialization algorithm using initial training samples of Digits \cite{uciml_digits} dataset (see Table \ref{tb:datasets} for details).
$\{\bm{P}_0, \bm{\beta}_0\}$ is obtained.
\textcircled{\scriptsize{3}} Compute training algorithm \textit{by one step} using online training samples. $\{\bm{P}_k, \bm{\beta}_k\}$ is obtained if $i = k$.
\textcircled{\scriptsize{4}} Generate 1,000 random training samples $\{\bm{x}, \bm{t}\}$ with uniform distribution of [0, 1].
Feed all the random samples into training algorithm of step = $k$
and measure the maximum and minimum values for each of $\{\bm{\gamma}_k^{(1)}, \ldots, \bm{\gamma}_k^{(10)}$, $\bm{P}_k$, $\bm{\beta}_k$, $\bm{e}_k, \bm{h}_k\}$.
\textcircled{\scriptsize{5}} Iterate 3-4 until all online training samples run out.

Figure \ref{fig:digits} shows the result.
We observed that all the intervals gradually converged or kept constant as $i$ proceeds.
Similar outcomes were observed on other datasets too (see Section \ref{sec:eval.hypothesis} for the entire result on multiple datasets).
From these outcomes, we make a hypothesis that $\bm{A}_i \in \{\bm{\gamma}_i^{(1)}, \ldots, \bm{\gamma}_i^{(10)}$, $\bm{P}_i$, $\bm{\beta}_i$, $\bm{e}_i, \bm{h}_i\}$
roughly satisfies $[\mathrm{min}(\bm{A}_1), \mathrm{max}(\bm{A}_1)] \supseteq [\mathrm{min}(\bm{A}_i), \mathrm{max}(\bm{A}_i)]$ for $2 \le i$,
in other words, the interval of $\bm{A}_1$ can be used as those of $\bm{A}_1, \ldots, \bm{A}_N$.
This hypothesis is verified in Section \ref{sec:eval.hypothesis}, using multiple datasets.

Based on the hypothesis we set $N = 1$ in training graph.
The interval analysis method for training graph is summarized as follows.
\begin{enumerate}
    \item Build training graph $\mathrm{T}(\bm{x}_0, \bm{t}_0, \bm{\alpha}, \bm{b}, \bm{P}_0, \bm{\beta}_0) \mapsto \{\bm{P}_1, \bm{\beta}_1\}$.
    \item Compute $\{\hat{\bm{\gamma}}_1^{(1)}, \ldots, \hat{\bm{\gamma}}_1^{(10)}$, $\hat{\bm{P}}_1$, $\hat{\bm{\beta}}_1$, $\hat{\bm{e}}_1$, $\hat{\bm{h}}_1\}$ using AA.
    The intervals are used as those of $\{\bm{\gamma}_i^{(1)}, \ldots, \bm{\gamma}_i^{(10)}$, $\bm{P}_i$, $\bm{\beta}_i$, $\bm{e}_i, \bm{h}_i\}$ ($i \ge 1$).
\end{enumerate}

\subsubsection{Division}\label{sec:method.train_analysis.div}
OS-ELM's training algorithm has a division $\frac{\bm{P}_{i-1}\bm{h}_i^T\bm{h}_i\bm{P}_{i-1}}{1+\bm{h}_i\bm{P}_{i-1}\bm{h}_i^T}$.
As mentioned in Section \ref{sec:prelim.aa}, the denominator $\gamma^{(5)}_i = 1 + \bm{h}_i\bm{P}_{i-1}\bm{h}_i^T$ must not take zero.
In the rest of this section $0 \notin \gamma^{(5)}_i$ is proven for $i \ge 1$.

\begin{theorem}
    $\bm{P}_{i-1}$ is positive-definite for $i \ge 1$.
\end{theorem}
\begin{proof}
    We first prove that $\bm{P}_0$ is positive-definite.
    \begin{itemize}
        \item $\bm{P}_0^{-1}$ is positive-semidefinite due to $\bm{u}\bm{P}_0^{-1}\bm{u}^T = \bm{u}\bm{H}_0^T\bm{H}_0\bm{u}^T = (\bm{u}\bm{H}_0^T) \cdot (\bm{u}\bm{H}_0^T)^T \ge 0$,
        where $\bm{u} \in \mathbb{R}^{1 \times \tilde{N}}$ represents an arbitrary vector.
        \item $\bm{P}_0^{-1} = \bm{H}_0^T\bm{H}_0$ is positive-definite since $\bm{H}_0^T\bm{H}_0$ is assumed to be a regular matrix in OS-ELM.
        \item $\bm{P}_0$ is positive-definite since the inverse of a positive-definite matrix is positive-definite.

    \end{itemize}
    Next, we prove that $\bm{P}_1$ is positive-definite.
    Equation \ref{eq:morrison} is derived by applying the sherman-morrison formula\footnote{
        $(\bm{V} + \bm{u}^T\bm{w})^{-1} = \bm{V}^{-1} - \frac{\bm{V}^{-1}\bm{u}^T\bm{w}\bm{V}^{-1}}{1 + \bm{w}\bm{V}^{-1}\bm{u}^T}$ ($\bm{V} \in \mathbb{R}^{k \times k}, \bm{u} \in \mathbb{R}^{1 \times k}, \bm{w} \in \mathbb{R}^{1 \times k}$, $k \in \mathbb{N}$).
    } to Equation \ref{eq:os_elm.train}.
    \begin{equation}\label{eq:morrison}
        \bm{P}_i = (\bm{P}_{i-1}^{-1} + \bm{h}_i^T\bm{h}_i)^{-1}
    \end{equation}
    \begin{itemize}
        \item $\bm{P}_1 = (\bm{P}_{0}^{-1} + \bm{h}_1^T\bm{h}_1)^{-1}$ holds by substituting $i = 1$.
        \item $\bm{h}_1^T\bm{h}_1$ is positive-semidefinite due to $\bm{u}\bm{h}_1^T\bm{h}_1\bm{u}^T = (\bm{u}\bm{h}_1^T) \cdot (\bm{u}\bm{h}_1^T)^T \ge 0$.
        \item $\bm{P}_1^{-1} = (\bm{P}_0^{-1} + \bm{h}_1^T\bm{h}_1)$ is positive-definite
        since it is the sum of a positive-definite matrix $\bm{P}_0^{-1}$ and a positive-semidefinite matrix $\bm{h}_1^T\bm{h}_1$.
        \item $\bm{P}_1$ is positive-definite since it is the inverse of a positive-definite matrix $\bm{P}_1^{-1}$.
    \end{itemize}
    By repeating the above logic, $\bm{P}_0, \ldots, \bm{P}_{i - 1}$ ($i \ge 1$) are all positive-definite.
\end{proof}
\begin{theorem}
    $\bm{h}_i\bm{P}_{i-1}\bm{h}_i^T > 0$ for $i \ge 1$.
\end{theorem}
\begin{proof}
    An $n \times n$ positive-definite matrix $\bm{V} \in \mathbb{R}^{n \times n}$ satisfies the following inequality.
    \begin{equation}
        \bm{u}\bm{V}\bm{u}^T > 0,
    \end{equation}
    where $\bm{u} \in \mathbb{R}^{1 \times n}$ represents an arbitrary vector.
    By applying this to $\bm{h}_i\bm{P}_{i-1}\bm{h}_i^T$,
    $\bm{h}_i\bm{P}_{i-1}\bm{h}_i^T > 0$ holds for $i \ge 1$,
    which guarantees $0 \notin 1 + \bm{h}_i\bm{P}_{i-1}\bm{h}_i^T \Leftrightarrow 0 \notin \gamma^{(5)}_i$ for $i \ge 1$.
\end{proof}
Note that $\mathrm{interval}(\hat{\gamma}^{(5)}_i)$ can include zero
because $\mathrm{interval}(\hat{\gamma}^{(5)}_i)$ can be wider than the true interval of $\gamma^{(5)}_i$.
To tackle this problem we propose to compute $\mathrm{min}(1, \mathrm{inf}(\hat{\gamma}^{(5)}_i))$ for the lower bound of $\hat{\gamma}^{(5)}_i$ instead of $\mathrm{inf}(\hat{\gamma}^{(5)}_i)$.
This trick prevents $\hat{\gamma^{(5)}_i}$ from including zero and
at the same time makes the interval close to the true interval.
Thanks to this trick OS-ELM's training algorithm can be safely represented in AA.

\subsection{Interval Analysis for Prediction Graph}\label{sec:method.predict_graph}
Prediction graph takes $\{\bm{x}, \bm{\beta}\}$ as input.
The interval of $\bm{\beta}$ should be that of $\bm{\beta}_i$ over $0 \le i \le N$, more specifically, $0 \le i \le 1$ ($N = 1$).
We propose to compute $\mathrm{min}(\mathrm{inf}(\hat{\beta}_{0[u,v]}), \mathrm{inf}(\hat{\beta}_{1[u,v]}))$ for the lower bound of $\beta_{[u,v]}$
and $\mathrm{max}(\mathrm{sup}(\hat{\beta}_{0[u,v]}), \mathrm{sup}(\hat{\beta}_{1[u,v]}))$ as the upper bound,
where $\beta_{[u,v]}$ represents the $uv$ element of $\bm{\beta}$.

\section{OS-ELM Core}\label{sec:impl}
We developed OS-ELM Core, a fixed-point IP core that implements OS-ELM algorithms, to verify the proposed interval analysis method.
All integer bit-widths of OS-ELM Core are parametrized, and the result of proposed interval analysis method is used as the arguments.
The PYNQ-Z1 FPGA board \cite{pynqz1} (280 BRAM blocks, 220 DSP slices, 106,400 flip-flops, and 53,200 LUT instances) is employed as the evaluation platform.

Figure \ref{fig:oselm_core} shows the block diagram of OS-ELM Core.
OS-ELM Core employs axi-stream protocol for input/output interface with 64-bit data width.
Training module executes OS-ELM's training algorithm
then updates $\bm{P}$ and $\bm{\beta}$ managed in parameter buffer.
Prediction module reads an input $\bm{x}$ from input buffer
and executes prediction algorithm.
The output of prediction module $\bm{y}$ is buffered in output buffer.
Both training and prediction modules use one adder and one multiplier
in a matrix product operation, and
one arithmetic unit (i.e. adder, multiplier, or divisor)
in a element-wise operation, regardless of the size of matrix,
to make hardware resource cost as small as possible.
All the arrays existing in OS-ELM Core are implemented with BRAM blocks (18 kb/block),
and all the fixed-point arithmetic units (i.e. adder, multiplier, and divisor) are with DSP slices.
\begin{figure}[b]
    \centering
    \includegraphics[width=80.0mm]{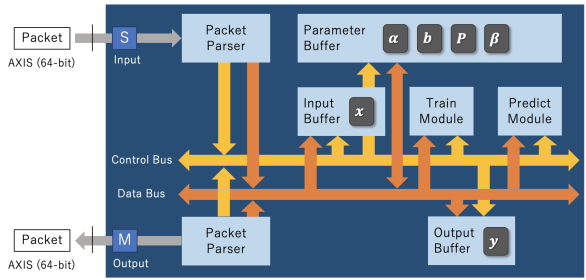}
    \caption{Block diagram of OS-ELM Core.}
    \label{fig:oselm_core}
\end{figure}

\begin{table}[t]
    \centering
    \caption{
        Classification datasets used in Section \ref{sec:eval}.
        ``Initial training samples' refers to the training samples used for computing $\{\bm{\beta}_0, \bm{P}_0\}$.
        ``Online training samples'' are the training samples for computing $\{\bm{\beta}_i, \bm{P}_i\}$ $(i \ge 1)$.
        ``Test samples'' are used to evaluate test accuracy and determine the number of hidden nodes.
        ``Features'' is the number of dimension of input $\bm{x}$.
        ``Classes'' corresponds to the number of output classes (= the number of dimension of output $\bm{y}$ and target $\bm{t}$).
        ``Model size'' column shows the model size $\{n, \tilde{N}, m\}$ for each dataset, where $n$, $\tilde{N}$, or $m$ represents
        the number of input, hidden, or output nodes.
    }
    \label{tb:datasets}
    \small
    \begin{tabular}{c|c|c|c|c|c|c}\hline\hline
    Name & Initial training samples & Online training samples & Test samples & Features & Classes & Model size\\\hline
    Digits \cite{uciml_digits} & 358 & 1,079 & 360 & 64 & 10 & \{64, 48, 10\}\\
    Iris \cite{uciml_iris} & 30 & 90 & 30 & 4 & 3 & \{4, 5, 3\}\\
    Letter \cite{uciml_letter} & 4,000 & 12,000 & 4,000 & 16 & 26 & \{16, 32, 26\}\\
    Credit \cite{uciml_credit} & 6,000 & 18,000 & 6,000 & 23 & 2 & \{23, 16, 2\}\\
    Drive \cite{uciml_drive} & 11,701 & 35,106 & 11,702 & 48 & 11 & \{48, 64, 11\}\\\hline
    \end{tabular}
\end{table}
\begin{table}[t]
    \centering
    \caption{Intervals obtained from simulation (sim) and the proposed interval analysis method (ours) for each dataset.}
    \label{tb:check}
    \scriptsize
    \begin{tabular}{c|c|c|c|c|c}\hline\hline
    & $\bm{\gamma}^{(1)}_i$ & $\bm{\gamma}^{(2)}_i$ & $\bm{\gamma}^{(3)}_i$ & $\bm{\gamma}^{(4)}_i$ & $\bm{\gamma}^{(5)}_i$\\\hline
    Digits (sim) & $[-0.642, 0.694]$ & $[-0.642, 0.694]$ & $[-0.446, 0.482]$ & $[0.371, 9.75]$ & $[1.37, 10.7]$\\
    Digits (ours) & $[-9.92e^{3}, 9.91e^{3}]$ & $[-9.26, 9.69]$ & $[-24.5, 27.8]$ & $[0.0, 1.46e^{3}]$ & $[1.0, 1.46e^{3}]$\\\hline
    Iris (sim) & $[-5.94, 5.85]$ & $[-5.94, 5.85]$ & $[-4.89, 35.3]$ & $[9.27e^{-3}, 3.24]$ & $[1.01, 4.24]$\\
    Iris (ours) & $[-1.55e^{3}, 1.55e^{3}]$ & $[-63.5, 19.1]$ & $[-388, 388]$ & $[0.0, 48.0]$ & $[1.0, 41.7]$\\\hline
    Letter (sim) & $[-6.72e^{-3}, 7.54e^{-3}]$ & $[-6.72e^{-3}, 7.54e^{-3}]$ & $[-5.06e^{-5}, 5.68e^{-5}]$ & $[2.79e^{-3}, 0.0397]$ & $[1.0, 1.04]$\\
    Letter (ours) & $[-0.301, 0.307]$ & $[-0.0593, 0.0785]$ & $[-2.42e^{-3}, 2.44e^{-3}]$ & $[0.0, 3.49]$ & $[1.0, 4.49]$\\\hline
    Credit (sim) & $[-0.115, 0.116]$ & $[-0.115, 0.116]$ & $[-8.36e^{-3}, 0.0135]$ & $[5.89e^{-3}, 0.253]$ & $[1.01, 1.25]$\\
    Credit (ours) & $[-32.9, 32.9]$ & $[-2.22, 3.25]$ & $[-0.589, 0.589]$ & $[0.0, 32.4]$ & $[1.0, 33.4]$\\\hline
    Drive (sim) & $[-6.97e^{5}, 6.92e^{5}]$ & $[-6.98e^{5}, 6.92e^{5}]$ & $[-3.71e^{11}, 4.87e^{11}]$ & $[5.26e^{4}, 4.72e^{6}]$ & $[5.26e^{4}, 4.72e^{6}]$\\
    Drive (ours) & $[-6.56e^{15}, 6.56e^{15}]$ & $[-1.33e^{7}, 1.56e^{7}]$ & $[-1.4e^{13}, 1.4e^{13}]$ & $[0.0, 1.55e^{9}]$ & $[1.0, 1.55e^{9}]$ \\\hline
    & $\bm{\gamma}^{(6)}_i$ & $\bm{\gamma}^{(7)}_i$ & $\bm{\gamma}^{(8)}_i$ & $\bm{\gamma}^{(9)}_i$ & $\bm{\gamma}^{(10)}_i$\\\hline
    Digits (sim) & $[-0.0447, 0.0472]$ & $[-0.102, 0.109]$ & $[-3.25, 3.29]$ & $[-3.0, 3.94]$ & $[-0.291, 0.306]$\\
    Digits (ours) & $[-25.8, 27.8]$ & $[-9.92e^{3}, 9.91e^{3}]$ & $[-12.1, 15.4]$ & $[-8.38, 9.0]$ & $[-8.93e^{4}, 8.93e^{4}]$\\\hline
    Iris (sim) & $[-1.32, 8.32]$ & $[-1.68, 1.67]$ & $[-1.24, 1.69]$ & $[-1.5, 2.12]$ & $[-2.1, 2.77]$\\
    Iris (ours) & $[-397, 397]$ & $[-1.55e^{3}, 1.55e^{3}]$ & $[-2.61, 2.3]$ & $[-2.3, 2.84]$ & $[-4.4e^{3}, 4.4e^{3}]$\\\hline
    Letter (sim) & $[-4.87e^{-5}, 5.46e^{-5}]$ & $[-6.47e^{-3}, 7.25e^{-3}]$ & $[-1.29, 1.03]$ & $[-0.869, 2.21]$ & $[-0.0104, 0.0129]$\\
    Letter (ours) & $[-2.84e^{-3}, 2.86e^{-3}]$ & $[-0.301, 0.307]$ & $[-3.11, 2.02]$ & $[-1.87, 3.31]$ & $[-1.01, 1.01]$\\\hline
    Credit (sim) & $[-7.11e^{-3}, 0.0115]$ & $[-0.0994, 0.0989]$ & $[-2.19, 3.9]$ & $[-3.89, 3.03]$ & $[-0.314, 0.245]$ \\
    Credit (ours) & $[-0.606, 0.606]$ & $[-32.9, 32.9]$ & $[-11.5, 10.7]$ & $[-6.25, 5.62]$ & $[-206, 206]$ \\\hline
    Drive (sim) & $[-1.36e^{5}, 1.65e^{5}]$ & $[-1.55, 1.39]$ & $[-962, 1.01e^{3}]$ & $[-1.01e^{3}, 970]$ & $[-345, 308]$\\
    Drive (ours) & $[-1.4e^{13}, 1.4e^{13}]$ & $[-6.56e^{15}, 6.56e^{15}]$ & $[-1e^{4}, 8.36e^{3}]$ & $[-3.42e^{3}, 3.44e^{3}]$ & $[-2.26e^{19}, 2.26e^{19}]$\\\hline
    & $\bm{P}_i$ & $\bm{\beta}_i$ & $\bm{e}_i$ & $\bm{h}_i$ & $\bm{y}$\\\hline
    Digits (sim) & $[-0.0544, 0.0705]$ & $[-0.351, 0.451]$ & $[-10.6, 9.15]$ & $[-10.0, 9.19]$ & $[-3.16, 3.25]$\\
    Digits (ours) & $[-27.4, 26.2]$ & $[-8.93e^{4}, 8.93e^{4}]$ & $[-23.1, 20.1]$ & $[-22.5, 20.8]$ & $[-3.39e^{7}, 3.39e^{7}]$\\\hline
    Iris (sim) & $[-1.72, 11.4]$ & $[-3.44, 5.32]$ & $[-2.44, 1.41]$ & $[-3.0, 2.21]$ & $[-1.23, 1.79]$\\
    Iris (ours) & $[-358, 435]$ & $[-4.4e^{3}, 4.4e^{3}]$ & $[-2.53, 1.58]$ & $[-3.1, 2.38]$ & $[-1.71e^{4}, 1.71e^{4}]$\\\hline
    Letter (sim) & $[-1.66e^{-3}, 2.45e^{-3}]$ & $[-0.34, 0.294]$ & $[-4.6, 5.33]$ & $[-4.86, 6.01]$ & $[-1.25, 1.18]$\\
    Letter (ours) & $[-9.2e^{-3}, 0.0126]$ & $[-1.35, 0.99]$ & $[-6.6, 7.8]$ & $[-6.87, 8.48]$ & $[-95.7, 95.3]$\\\hline
    Credit (sim) & $[-0.0649, 0.115]$ & $[-1.83, 1.38]$ & $[-4.66, 5.5]$ & $[-5.55, 6.22]$ & $[-2.18, 3.77]$ \\
    Credit (ours) & $[-0.625, 1.05]$ & $[-204, 208]$ & $[-8.29, 9.66]$ & $[-9.19, 10.4]$ & $[-1.09e^{4}, 1.09e^{4}]$ \\\hline
    Drive (sim) & $[-1.4e^{5}, 1.7e^{5}]$ & $[-317, 318]$ & $[-9.9, 7.42]$ & $[-9.35, 8.29]$ & $[-1.21e^{3}, 318]$ \\
    Drive (ours) & $[-1.4e^{13}, 1.4e^{13}]$ & $[-2.26e^{19}, 2.26e^{19}]$ & $[-18.3, 16.8]$ & $[-17.7, 16.0]$ & $[-1.06e^{22}, 1.06e^{22}]$\\\hline
    \end{tabular}
\end{table}
\section{Evaluation}\label{sec:eval}
In this section we evaluate the proposed interval analysis method.
All the experiments here were executed on a server machine (Ubuntu 20.04, Intel Xeon E5-1650 3.60GHz, DRAM 64GB, SSD 500GB).
Table \ref{tb:datasets} lists the classification datasets used for evaluation of our method.
For all the datasets, the intervals of input $\bm{x}$ and target $\bm{t}$ are normalized into $[0, 1]$.
Parameters $\bm{b}$ and $\bm{\alpha}$ are randomly generated with the uniform distribution of $[-1, 1]$.
The model size for each dataset is shown in ``Model Size'' column.
The number of hidden nodes is set to the number that performed the best test accuracy in a given search space;
search spaces for Digits, Iris, Letter, Credit, and Drive are \{32, 48, 64, 96, 128\}, \{3, 4, 5, 6, 7\}, \{8, 16, 32, 64, 128\},
\{4, 8, 16, 32, 64\}, and \{32, 64, 96, 128\} respectively.
\begin{table}[b]
    \centering
    \caption{
        The ``Ops'' column shows the total number of arithmetic operations,
        and the ``Overflow/Underflow'' column shows the number of overflow or underflows that happened during the experiment.
        The rate of overflow/underflows is written in ().}
    \label{tb:rate}
    \begin{tabular}{c|c|c}\hline\hline
        & Ops & Overflow/Underflow \\\hline
        Digits (sim) & \multirow{2}{*}{5,512,688,688} & 0 \\
        Digits (ours) & & 0 \\\hline
        Iris (sim) & \multirow{2}{*}{4,714,041} & 197,342 (4.19\%) \\
        Iris (ours) & & 0 \\\hline
        Letter (sim) & \multirow{2}{*}{17,793,216,000} & 0 \\
        Letter (ours) & & 0 \\\hline
        Credit (sim) & \multirow{2}{*}{11,039,328,000} & 0 \\
        Credit (ours) & & 0 \\\hline
        Drive (sim) & \multirow{2}{*}{187,259,827,356} & 5,467,945,469 (2.92\%) \\
        Drive (ours) & & 0 \\\hline
    \end{tabular}
\end{table}

\subsection{Optimization Result}\label{sec:eval.check}
In this section we first show the result of the proposed interval analysis method for each dataset,
comparing with an ordinary simulation-based interval analysis method.
Here is a brief introduction of the simulation method:
\textcircled{\scriptsize{1}} Implement OS-ELM's initialization, prediction, and training algorithms in double-precision format.
\textcircled{\scriptsize{2}} Execute initialization algorithm using initial training samples. $\{\bm{P}_0, \bm{\beta}_0\}$ is obtained.
\textcircled{\scriptsize{3}} Execute training algorithm by one step using online training samples. $\{\bm{P}_k, \bm{\beta}_k\}$ is obtained if $i = k$.
\textcircled{\scriptsize{4}} Generate 1,000 random training samples $\{\bm{x}, \bm{t}\}$ with uniform distribution of $[0, 1]$.
\textcircled{\scriptsize{5}} Feed all the random samples into training algorithm of step = $k$ and measure the values of $\{\bm{\gamma}_k^{(1)}, \ldots, \bm{\gamma}_k^{(10)}, \bm{P}_k, \bm{\beta}_k, \bm{e}_k, \bm{h}_k\}$.
\textcircled{\scriptsize{6}} Feed all the random samples into prediction algorithm and measure the values of $\bm{y}$.
\textcircled{\scriptsize{7}} Repeat 3-6 until all online training samples run out.

Table \ref{tb:check} shows the intervals obtained from the simulation method (sim)
and those from the proposed method (ours).
All the intervals obtained from our method cover the corresponding simulated interval.
Note that the simulated interval of $\bm{\gamma}_i^{(5)} = 1 + \bm{h}_i\bm{P}_{i-1}\bm{h}_i^T$ satisfies $\bm{\gamma}_i^{(5)} > 1$,
which is consistent with the theorem proven in Section \ref{sec:method.train_analysis.div}.

\subsection{Rate of Overflow/Underflows}\label{sec:eval.rate}
This section compares the simulation method introduced in Section \ref{sec:eval.check} and
the proposed method in terms of the rate of overflow/underflows, using OS-ELM Core.
The experimental procedure is as follows:
\textcircled{\scriptsize{1}} Execute the simulation method and convert the result into integer bit-widths using Equation \ref{eq:optim}
(an extra bit was added to each bit-width to reduce overflow/underflows).
\textcircled{\scriptsize{2}} Execute the proposed method and convert the result into bit-widths.
\textcircled{\scriptsize{3}} Synthesize two OS-ELM Cores using the bit-widths obtained from 1 and 2.
\textcircled{\scriptsize{4}} Execute training by one step in both OS-ELM Cores using online training samples.
\textcircled{\scriptsize{5}} Generate 250 random training samples $\{\bm{x}, \bm{t}\}$ with uniform distribution of $[0, 1]$.
\textcircled{\scriptsize{6}} Feed all the random samples into the training module and the prediction module for each OS-ELM Core and check the number of overflow/underflows that arose.
\textcircled{\scriptsize{7}} Repeat 4-6 until all online training samples run out.

The result is shown in Table \ref{tb:rate}.
The simulation method caused no overflow or underflows in three datasets out of five, however,
it suffered from as many overflow/underflows as 2.92 $\sim$ 4.19\% in the other two datasets,
where a few overflow/underflows arose in an early training step and were propagated to subsequent steps,
resulting in a drastic increase in overflow/underflows.
This cannot be perfectly prevented as long as a random exploration is taken in interval analysis.
The proposed method, on the other hand, encountered totally no overflow or underflows
as it analytically derives upper and lower bounds of variables and computes sufficient bit-widths where no overflow or underflows can happen.
Although the proposed method produces some redundant bits and it results in a larger area size (see Section \ref{sec:eval.area}),
it safely realizes an overflow/underflow-free fixed-point OS-ELM circuit.

\subsection{Verification of Hypothesis}\label{sec:eval.hypothesis}
\begin{figure*}[t]
    \begin{tabular}{c}
        \begin{minipage}{0.25\hsize}
            \centering
            \includegraphics[width=170.0mm]{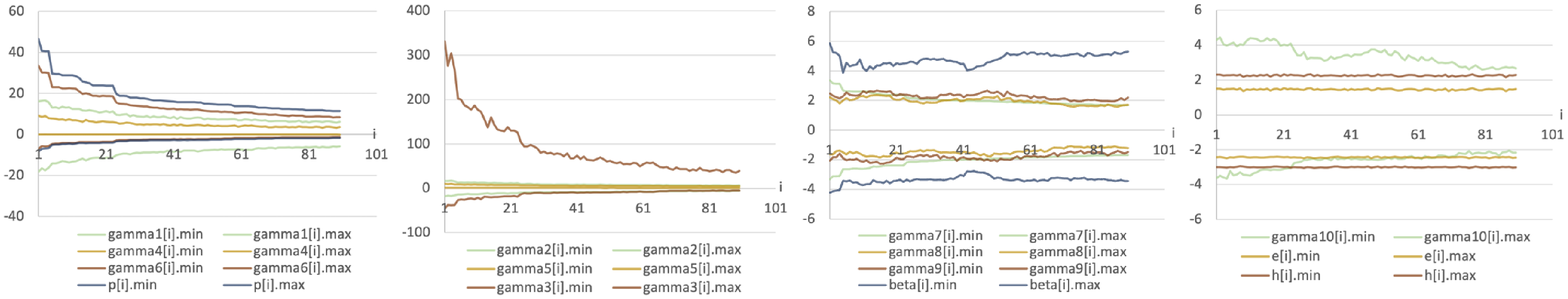}
        \end{minipage} \\
        \begin{minipage}{0.25\hsize}
            \centering
            \includegraphics[width=170.0mm]{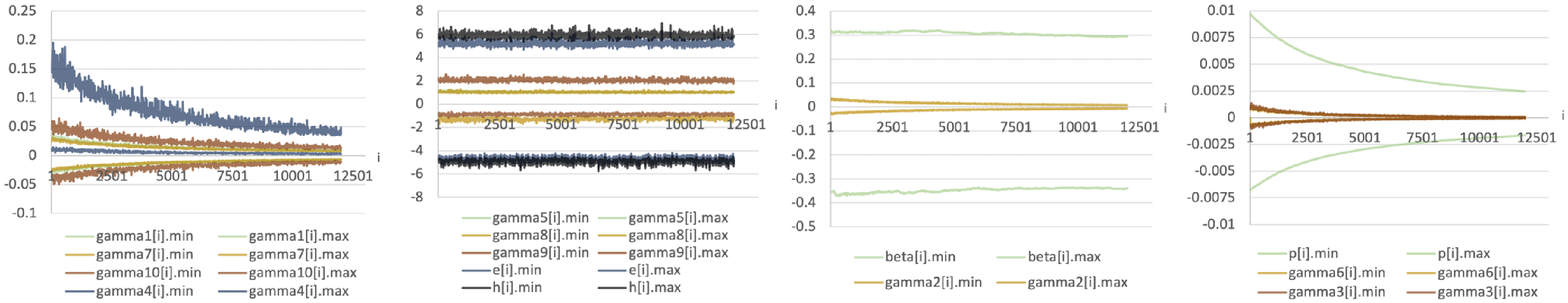}
        \end{minipage} \\
        \begin{minipage}{0.25\hsize}
            \centering
            \includegraphics[width=170.0mm]{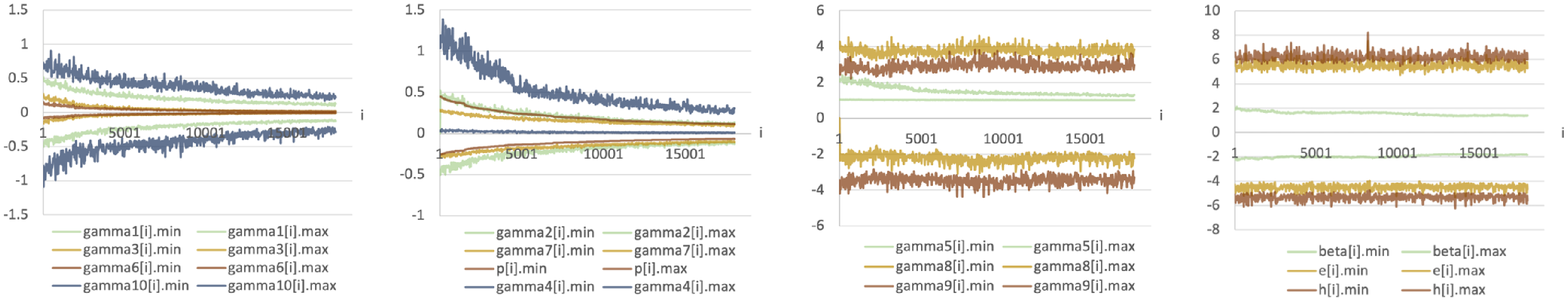}
        \end{minipage} \\
        \begin{minipage}{0.25\hsize}
            \centering
            \includegraphics[width=170.0mm]{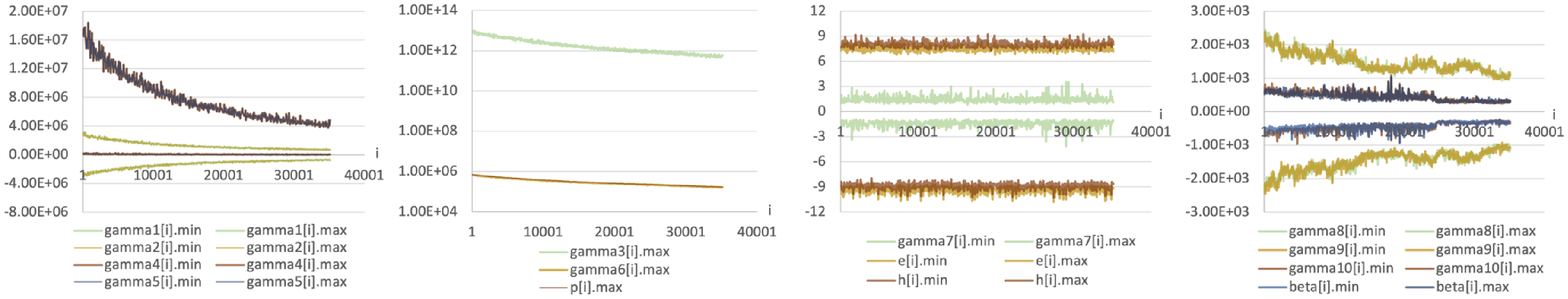}
        \end{minipage}
    \end{tabular}
    \caption{Observed intervals of $\{\bm{\gamma}^{(1)}_i, \ldots, \bm{\gamma}^{(10)}_i, \bm{P}_i, \bm{\beta}_i, \bm{e}_i, \bm{h}_i\}$ on Iris (top row), Letter (2nd row), Credit (3rd row), and Drive (bottom row), respectively.}
    \label{fig:hypothesis}
\end{figure*}
Figure \ref{fig:hypothesis} shows the entire result of the experiment described in Section \ref{sec:method.train_analysis.n}.
We observed similar outcomes to Figure \ref{fig:digits} for all the datasets,
which supports our hypothesis that $\bm{A}_i \in \{\bm{\gamma}_i^{(1)}, \ldots, \bm{\gamma}_i^{(10)}$, $\bm{P}_i$, $\bm{\beta}_i$, $\bm{e}_i, \bm{h}_i\}$
roughly satisfies $[\mathrm{min}(\bm{A}_1), \mathrm{max}(\bm{A}_1)] \supseteq [\mathrm{min}(\bm{A}_i), \mathrm{max}(\bm{A}_i)]$ for $2 \le i$.

In iterative learning algorithms it is known that learning parameters ($\bm{\beta}_i$ and $\bm{P}_i$ in the case of OS-ELM) gradually converge to some values as training proceeds.
We consider that this numerical property resulted in the convergence of the dynamic ranges of $\bm{\beta}_i$ and $\bm{P}_i$ as observed in Figure \ref{fig:hypothesis},
then it tightened the dynamic ranges of other variables (e.g. $\bm{\gamma}_i^{(1)}, \ldots, \bm{\gamma}_i^{(10)}$) too, as a side-effect via enormous number of multiplications existing in the OS-ELM algorithm.
We plan to investigate the hypothesis either by deriving an analytical proof or using a larger dataset in the future work.

\subsection{Area Cost}\label{sec:eval.area}
In this section the proposed method is evaluated in terms of area cost.
We refer to BRAM utilization of OS-ELM Core as ``area cost'',
considering that all the arrays in OS-ELM Core are implemented with BRAM blocks (i.e. the bottleneck of area cost is BRAM utilization).
The proposed method is compared with the simulation method introduced in Section \ref{sec:eval.check}
to clarify how much additional area cost arises to guarantee OS-ELM Core being overflow/underflow-free.
The experimental procedure is as follows:
\textcircled{\scriptsize{1}} Convert the simulation result into integer bit-widths using Equation \ref{eq:optim} and synthesize OS-ELM Core with the optimized bit-widths.
\textcircled{\scriptsize{2}} Execute the proposed interval analysis method. Convert the result into integer bit-widths and synthesize OS-ELM Core.
\textcircled{\scriptsize{3}} Check the BRAM utilizations of our method and the simulation method.
\textcircled{\scriptsize{4}} Repeat 1-3 for all the datasets.
\begin{figure}[t]
    \centering
    \includegraphics[width=80.0mm]{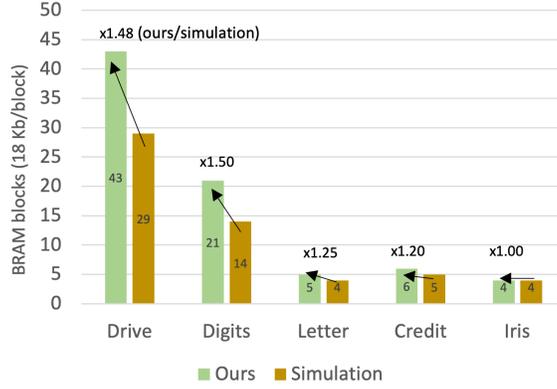}
    \caption{
        Comparison of area cost.
        The green bar represents the BRAM utilization of our method and the brown bar is of the simulation method.}
    \label{fig:area}
\end{figure}

The experimental result is shown in Figure \ref{fig:area}.
Our method requires 1.0x - 1.5x more BRAM blocks
to guarantee that OS-ELM Core never encounter overflow and underflow, compared to the simulation method.

Remember that a multiplication in AA causes overestimation of interval;
there should be a strong correlation between the additional area cost (i.e. simulation - ours)
and the number of multiplications in OS-ELM's training and prediction algorithms.
\begin{equation}\label{eq:mult}
    \mathrm{M}(n, \tilde{N}, m) = 4\tilde{N}^{2} + (3m + n + 1)\tilde{N}
\end{equation}
$\mathrm{M}(n, \tilde{N}, m)$ calculates the total number of multiplications in OS-ELM's training and prediction algorithms,
where $n$, $\tilde{N}$, or $m$ is the number of input, hidden, or output nodes, respectively.
Equation \ref{eq:mult} shows that $\tilde{N}$ has the largest impact on additional area cost,
which is consistent with the result that 2.0x more additional area cost was observed in Drive
compared to Digits, with fewer inputs nodes (Drive: 48, Digits: 64), more hidden nodes (Drive: 64, Digits: 48),
and almost the same number of hidden nodes (Drive: 11, Digits: 10).
We conclude that the proposed method is highly effective especially when the model size is small,
and that the number of hidden nodes has the strongest impact on additional area cost.

\section{Conclusion}\label{sec:conc}
In this paper we proposed an overflow/underflow-free bit-width optimization method for fixed-point OS-ELM digital circuits.
In the proposed method affine arithmetic is used to estimate the intervals of intermediate variables
and compute the optimal number of integer bits that never cause overflow and underflow.
We clarified two critical problems in realizing the proposed method:
(1) OS-ELM's training algorithm is an iterative algorithm and the computation graph grows endlessly,
which makes interval analysis infeasible in affine arithmetic.
(2) OS-ELM's training algorithm has a division operation
and if the denominator can take zero
OS-ELM can not be represented in affine arithmetic.

We proposed an empirical solution to prevent the computation graph from growing endlessly, based on simulation results.
We also analytically proved that the denominator does not take zero at any training step,
and proposed a mathematical trick based of the proof to safely represent OS-ELM in affine arithmetic.
Experimental results confirmed that no underflow/overflow occurred in our method on multiple datasets.
Our method realized overflow/underflow-free OS-ELM digital circuits
with 1.0x - 1.5x more area cost compared to the baseline simulation method
where overflow or underflow can happen.
\begin{table}[b]
    \centering
    \caption{
        Notation rules in this paper.
    }
    \label{tb:notation}
    \begin{tabular}{c|c}\hline\hline
    Notation & Description \\\hline
    $x$ (\textit{italic}) & Scaler. \\
    $\hat{x}$ & Affine form of $x$. \\
    $\bm{x}$ (\textbf{\textit{bold italic}}) & Vector or matrix. \\
    $\hat{\bm{x}}$ & \begin{tabular}{c} Affine form of $\bm{x}$ \\ (see Equation \ref{eq:affine.matrix_interval} for details). \end{tabular}\\
    $x_{[u,v]}$ & $uv$ element of $\bm{x}$. \\
    $\hat{x}_{[u,v]}$ & Affine form for the $uv$ element of $\bm{x}$.\\
    $\mathrm{f}$ (upright) & Function (e.g. $\mathrm{G}, \mathrm{sup}, \mathrm{inf}, \mathrm{interval}$).  \\\hline
    \end{tabular}
\end{table}
\begin{table*}[t]
    \centering
    \caption{
        Description of variables that appear in this paper.
        The characters used for these variables are the same as the ones used in \cite{os_elm}.
    }
    \label{tb:desc}
    \small
    \begin{tabular}{c|c}\hline\hline
    Variable & Description \\\hline
    $n, \tilde{N}, m \in \mathbb{N}$ & Number of input, hidden, or output nodes of OS-ELM. \\
    $\bm{\alpha} \in \mathbb{R}^{n \times \tilde{N}}$ & Non-trainable weight matrix connecting the input and hidden layers, which is initialized with random values. \\
    $\bm{\beta} \in \mathbb{R}^{\tilde{N} \times m}$ & Trainable weight matrix connecting the hidden and output layers. \\
    $\bm{P} \in \mathbb{R}^{\tilde{N} \times \tilde{N}}$ & Trainable intermediate weight matrix for training $\bm{\beta}$. \\
    $\bm{b} \in \mathbb{R}^{1 \times \tilde{N}}$ & Non-trainable bias vector of the hidden layer, which is initialized with random values. \\
    $\mathrm{G}$ & Activation function applied to the hidden layer output. \\
    $\bm{x} \in \mathbb{R}^{1 \times n}$ & Input vector. \\
    $\bm{t} \in \mathbb{R}^{1 \times m}$ & Target vector. \\
    $\bm{y} \in \mathbb{R}^{1 \times m}$ & Output vector. \\
    $\bm{h} \in \mathbb{R}^{1 \times \tilde{N}}$ & Output vector of the hidden layer (after activation). \\
    $\bm{e} \in \mathbb{R}^{1 \times \tilde{N}}$ & Output vector of the hidden layer (before activation). \\
    $\bm{X} \in \mathbb{R}^{k \times n}$ & Input matrix of batch size = $k$ ($k \in \mathbb{N}$). \\
    $\bm{T} \in \mathbb{R}^{k \times m}$ & Target matrix of batch size = $k$. \\
    $\bm{Y} \in \mathbb{R}^{k \times m}$ & Output matrix of batch size = $k$. \\
    $\bm{H} \in \mathbb{R}^{k \times \tilde{N}}$ & Output matrix of the hidden layer with batch size = $k$ (after activation). \\
    $\bm{\gamma}^{(1)}, \ldots,\bm{\gamma}^{(10)}$ & Intermediate variables that appear in OS-ELM's training algorithm. \\\hline
    \end{tabular}
\end{table*}


\bibliographystyle{unsrt}

\begin{thebibliography}{10}
  \bibitem{os_elm_fpga_tsukada_tc}
  M.~Tsukada, M.~Kondo, and H.~Matsutani.
  \newblock A {N}eural {N}etwork-based {O}n-device {L}earning {A}nomaly
    {D}etector for {E}dge {D}evices.
  \newblock {\em IEEE Transactions on Computers}, 69(7):1027--1044, Jul 2020.
  
  \bibitem{os_elm}
  N.Y. Liang, G.B. Huang, P.~Saratchandran, and N.~Sundararajan.
  \newblock A {F}ast and {A}ccurate {O}nline {S}equential {L}earning {A}lgorithm
    for {F}eedforward {N}etworks.
  \newblock {\em IEEE Transactions on Neural Networks}, 17(6):1411--1423, Nov
    2006.
  
  \bibitem{os_elm_fpga_tsukada}
  M.~Tsukada, M.~Kondo, and H.~Matsutani.
  \newblock O{S}-{E}{L}{M}-{F}{P}{G}{A}: {A}n {F}{P}{G}{A}-{B}ased {O}nline
    {S}equential {U}nsupervised {A}nomaly {D}etector.
  \newblock In {\em Proceedings of the International European Conference on
    Parallel and Distributed Computing Workshops}, pages 518--529, Aug 2018.
  
  \bibitem{os_elm_fpga_villora}
  J.V.F. Villora, A.R. {Mu\~{n}oz}, M.B. Mompean, J.B. Aviles, and J.F.G.
    Martinez.
  \newblock Moving {L}earning {M}achine towards {F}ast {R}eal-{T}ime
    {A}pplications: {A} {H}igh-{S}peed {F}{P}{G}{A}-{B}ased {I}mplementation of
    the {O}{S}-{E}{L}{M} {T}raining {A}lgorithm.
  \newblock {\em Electronics}, 7(11):1--23, Nov 2018.
  
  \bibitem{os_elm_fpga_safaei}
  A.~Safaei, Q.M.J. Wu, T.~Akilan, and Y.~Yang.
  \newblock System-on-a-{C}hip ({S}o{C})-based {H}ardware {A}cceleration for an
    {O}nline {S}equential {E}xtreme {L}earning {M}achine ({O}{S}-{E}{L}{M}).
  \newblock {\em IEEE Transactions on Computer-Aided Design of Integrated
    Circuits and Systems (Early Access)}, Oct 2018.
  
  \bibitem{aawlo_lee}
  D.U. Lee, A.A. Gaffer, R.C.C Cheung, O.~Mencer, W.~Luk, and G.A.
    Constantinides.
  \newblock Accuracy-{G}uaranteed {B}it-{W}idth {O}ptimization.
  \newblock {\em IEEE Transactions on Computer-Aided Design of Integrated
    Circuits and Systems}, 25(10):1990--2000, Oct 2006.
  
  \bibitem{smtwlo_kinsman}
  A.~Kinsman and N.~Nicolici.
  \newblock Bit-{W}idth {A}llocation for {H}ardware {A}ccelerators for
    {S}cientific {C}omputing {U}sing {S}{A}{T}-{M}odulo {T}heory.
  \newblock {\em IEEE Transactions on Computer-Aided Design of Integrated
    Circuits and Systems}, 29(3):405--413, Mar 2010.
  
  \bibitem{hrwlo_boland}
  D.~Boland and G.~Constantinides.
  \newblock Bounding {V}ariable {V}alues and {R}ound-{O}ff {E}ffects {U}sing
    {H}andelman {R}epresentations.
  \newblock {\em IEEE Transactions on Computer-Aided Design of Integrated
    Circuits and Systems}, 30(11):1691--1704, Nov 2011.
  
  \bibitem{aa}
  J.~Stolfi and L.~Figueiredo.
  \newblock Self-{V}alidated {N}umerical {M}ethods and {A}pplications, 1997.
  
  \bibitem{elm}
  G.B. Huang, Q.Y. Zhu, and C.K. Siew.
  \newblock Extreme {L}earning {M}achine: {A} {N}ew {L}earning {S}cheme of
    {F}eedforward {N}eural {N}etworks.
  \newblock In {\em Proceedings of the International Joint Conference on Neural
    Networks}, pages 985--990, Jul 2004.
  
  \bibitem{fixed_point_refine}
  D.~Menard, G.~Caffarena, J.~Antonio, A.~Lopez, D.~Novo, and O.~Sentieys.
  \newblock {\em Fixed-point refinement of digital signal processing systems},
    pages 1--37.
  \newblock The Institution of Engineering and Technology, May 2019.
  
  \bibitem{wlo_dynamic_cmar}
  R.~Cmar, L.~Rijnders, P.~Schaumont, S.~Vernalde, and I.~Bolsens.
  \newblock A methodology and design environment for {D}{S}{P} {A}{S}{I}{C} fixed
    point refinement.
  \newblock In {\em Design, Automation and Test in Europe Conference and
    Exhibition}, pages 271--276, Mar 1999.
  
  \bibitem{wlo_dynamic_gaffar}
  A.~Gaffar, O.~Mencer, and W.~Luk.
  \newblock Unifying bit-width optimisation for fixed-point and floating-point
    designs.
  \newblock In {\em The Annual IEEE Symposium on Field-Programmable Custom
    Computing Machines}, pages 79--88, Apr 2004.
  
  \bibitem{wlo_dynamic_keding}
  H.~Keding, M.~Willems, and H.~Meyr.
  \newblock Fridge: a fixed-point design and simulation environment.
  \newblock In {\em Design, Automation and Test in Europe Conference and
    Exhibition}, pages 429--435, Feb 1998.
  
  \bibitem{wlo_dynamic_shi}
  C.~Shi and R.~Brodersen.
  \newblock Automated fixed-point data-type optimization tool for signal
    processing and communication systems.
  \newblock In {\em Design Automation Conference}, pages 478--483, July 2004.
  
  \bibitem{aawlo_cong}
  J.~Cong, K.~Gururaj, B.~Liu, C.~Liu, Z.~Zhang, S.~Zhou, and Y.~Zou.
  \newblock Evaluation of {S}tatic {A}nalysis {T}echniques for {F}ixed-{P}oint
    {P}recision {O}ptimization.
  \newblock In {\em Proceedings of the IEEE Symposium on Field Programmable
    Custom Computing Machines}, pages 231--234, Apr 2009.
  
  \bibitem{aawlo_vakili}
  S.~Vakili, J.M.P Langlois, and G.~Bois.
  \newblock Enhanced {P}recision {A}nalysis for {A}ccuracy-{A}ware {B}it-{W}idth
    {O}ptimization {U}sing {A}ffine {A}rithmetic.
  \newblock {\em IEEE Transactions on Computer-Aided Design of Integrated
    Circuits and Systems}, 32(12):1853--1865, Dec 2013.
  
  \bibitem{ia}
  R.~Moore.
  \newblock Interval {A}nalysis.
  \newblock {\em Science}, 158(3799):365--365, Oct 1967.
  
  \bibitem{aawlo_fang}
  C.F. Fang, R.A. Rutenbar, and T.~Chen.
  \newblock Fast, accurate static analysis for fixed-point finite-precision
    effects in {D}{S}{P} designs.
  \newblock In {\em Proceedings of the International Conference on Computer Aided
    Design}, pages 1--8, Nov 2003.
  
  \bibitem{aawlo_pu}
  Y.~Pu and Y.~Ha.
  \newblock An automated, efficient and static bit-width optimization methodology
    towards maximum bit-width-to-error tradeoff with affine arithmetic model.
  \newblock In {\em Proceedings of the Asia and South Pacific Conference on
    Design Automation}, pages 886--891, Jan 2006.
  
  \bibitem{aawlo_wang}
  S.~Wang and X.~Qing.
  \newblock A {M}ixed {I}nterval {A}rithmetic/{A}ffine {A}rithmetic {A}pproach
    for {R}obust {D}esign {O}ptimization {W}ith {I}nterval {U}ncertainty.
  \newblock {\em Journal of Mechanical Design}, 138(4):041403--1--041403--10, Apr
    2016.
  
  \bibitem{aawlo_bellal}
  R.~Bellal, E.~Lamini, H.~Belbachir, S.~Tagzout, and A.~Belouchrani.
  \newblock Improved {A}ffine {A}rithmetic-{B}ased {P}recision {A}nalysis for
    {P}olynomial {F}unction {E}valuation.
  \newblock {\em IEEE Transactions on Computers}, 68(5):702--712, May 2019.
  
  \bibitem{uciml_digits}
  E.~Alpaydin and C.~Kaynak.
  \newblock Optical {R}ecognition of {H}andwritten {D}igits {D}ata {S}et.
  \newblock
    \url{https://archive.ics.uci.edu/ml/datasets/Optical+Recognition+of+Handwritten+Digits},
    1998.
  
  \bibitem{pynqz1}
  Digilent {P}{Y}{N}{Q}-{Z}1.
  \newblock
    \url{https://japan.xilinx.com/products/boards-and-kits/1-hydd4z.html}.
  
  \bibitem{uciml_iris}
  R.~Fisher.
  \newblock Iris {D}ata {S}et.
  \newblock \url{http://archive.ics.uci.edu/ml/datasets/Iris/}, 1936.
  
  \bibitem{uciml_letter}
  D.~Slate.
  \newblock Letter {R}ecognition {D}ata {S}et.
  \newblock \url{https://archive.ics.uci.edu/ml/datasets/Letter+Recognition},
    1890.
  
  \bibitem{uciml_credit}
  I.~Yeh.
  \newblock Default of {C}redit {C}ard.
  \newblock
    \url{https://archive.ics.uci.edu/ml/datasets/default+of+credit+card+clients},
    2016.
  
  \bibitem{uciml_drive}
  M.~Bator.
  \newblock Sensorless {D}rive {D}iagnosis.
  \newblock
    \url{https://archive.ics.uci.edu/ml/datasets/dataset+for+sensorless+drive+diagnosis},
    2015.
\end{thebibliography}

\end{document}